\documentclass[a4paper,english,cleveref, autoref, thm-restate]{lipics-v2021}
\usepackage{todonotes}
\usepackage{tikz,url}
\usepackage{tikz-cd}
\usetikzlibrary{arrows,automata,decorations}
\usepackage{amsmath,amssymb,amsthm}
\usepackage{epsfig}
\usepackage{mathrsfs}
\usepackage{tikzit}

\tikzstyle{white vertex}=[fill=white, draw=black, shape=circle, thick, inner sep=1pt, minimum size=11pt]
\tikzstyle{green vertex}=[fill=green, draw=black, shape=circle, thick, inner sep=1pt, minimum size=11pt]
\tikzstyle{cyan vertex}=[fill=cyan, draw=black, shape=circle, thick, inner sep=1pt, minimum size=11pt]
\tikzstyle{magenta vertex}=[fill=magenta, draw=black, shape=circle, thick, inner sep=1pt, minimum size=11pt]
\tikzstyle{empty node}=[fill=white, shape=circle, inner sep=1pt]

\tikzstyle{filled path}=[-, fill=lipicsLightGray, thick, fill opacity=0.7]
\tikzstyle{thick edge}=[-, thick]

\usepackage{macros}

\title{A Simplicial Model for $\KBfour$: Epistemic Logic with Agents that May Die}

\author{\'Eric Goubault}{LIX, CNRS, \'Ecole Polytechnique, Institut Polytechnique de Paris, Paris, France}{eric.goubault@polytechnique.edu}{}{}

\author{J\'er\'emy Ledent}{MSP Group, University of Strathclyde, Glasgow, Scotland}{jeremy.ledent@strath.ac.uk}{https://orcid.org/ 0000-0001-7375-4725}{}

\author{Sergio Rajsbaum}{UNAM, Mexico City, Mexico}{rajsbaum@im.unam.mx}{https://orcid.org/0000-0002-0009-5287}{This work was partially done while this author was at \'Ecole Polytechnique, France}

\titlerunning{A General Epistemic Logic Approach to Distributed Tasks}
\authorrunning{\'E. Goubault, J. Ledent and S. Rajsbaum}

\Copyright{\'Eric Goubault, J\'er\'emy Ledent and Sergio Rajsbaum}



\category{} 

\relatedversion{} 

\funding{Supported by UNAM-PAPIIT IN106520}

\nolinenumbers 

\hideLIPIcs  

\EventEditors{Petra Berenbrink and Benjamin Monmege}
\EventNoEds{2}
\EventLongTitle{39th International Symposium on Theoretical Aspects of Computer Science (STACS 2022)}
\EventShortTitle{STACS 2022}
\EventAcronym{STACS}
\EventYear{2022}
\EventDate{March 15--18, 2022}
\EventLocation{Marseille, France}
\EventLogo{}
\SeriesVolume{219}
\ArticleNo{31}

\begin{document}
\maketitle

\begin{abstract}
The standard semantics of multi-agent epistemic logic~$\Sfive$ is based on Kripke models whose accessibility relations are reflexive, symmetric and transitive.
This one dimensional structure contains implicit higher-dimensional information beyond pairwise interactions, that we formalized as pure simplicial models in a previous work in \emph{Information and Computation} 2021~\cite{gandalf-journal}. 
Here we extend the theory to encompass  simplicial models  that are not necessarily pure.
The corresponding class of Kripke models are those where the accessibility relation is symmetric and transitive, but might not be reflexive. Such models correspond to the  epistemic logic $\KBfour$.
Impure simplicial models arise in situations where two possible worlds may not have the same set of agents.
We illustrate it with distributed computing examples of synchronous systems where processes may crash.
\end{abstract}

\section{Introduction}

A very successful research programme of using  epistemic logic to reason about multi-agent systems began in the early 1980's showing the fundamental role of notions such as common knowledge~\cite{fagin,Moses2016}. The semantics used is the one of ``normal modal logics'', based on the classic \emph{possible worlds} relational structure  developed by Rudolf Carnap, Stig Kanger, Jakko Hintikka and Saul Kripke in the late 1950's and early 1960's. 

\subparagraph{From global states to local states.}
The intimate relationship between distributed computing and algebraic topology discovered in 1993~\cite{BorowskyG93,HS99,SaksZ00} showed the importance of moving from using worlds as the primary object, to \emph{perspectives} about possible worlds.
After all, what exists in many distributed systems is only the local states of the agents and events observable within the system.

Taking local states as the main notion led to the study of distributed systems based on geometric structures called \emph{simplicial complexes}.
In this context, a simplicial complex is constructed using the local states as vertices and the global states as simplexes.
While the solvability of some distributed tasks such as \emph{consensus} depends only on the one-dimensional (graph) connectivity of 
global states, 
the solvability of other tasks, most notably $k$-\emph{set agreement},
depends on the higher-dimensional connectivity of the simplicial complex of local states.
See~\cite{herlihyetal:2013} for an overview of the topological theory of distributed computability.

\subparagraph{Pure simplicial model semantics~\cite{gandalf-journal}.}
From the very beginning~\cite{SaksZ00}, distributed computer scientists have used the word ``knowledge'' informally to explain their use of simplicial complexes.
However, a formal link with epistemic logic was established only recently~\cite{gandalf-journal}.

The idea is to replace the usual one-dimensional Kripke models by a new class of models based on simplicial complexes, called \emph{simplicial models}.
In~\cite{gandalf-journal}, we focused on modelling the standard multi-agent epistemic logic, $\Sfive$.
In this setting, a core assumption is that the same set of~$n$ agents always participate in every possible world.
Because of this, all the facets of the simplicial model are of the same dimension.
Such models are called \emph{pure} simplicial models.
With this restriction, we showed that the class of pure simplicial models is equivalent to the usual class of $\Sfive$ Kripke models.

Using pure simplicial models, we provided epistemic logic tools to reason about solvability of distributed tasks such as consensus and approximate agreement.
In subsequent work, we also studied the equality negation task, explored bisimilarity of pure simplicial models,  and connections with covering spaces~\cite{Ditmarsch2020KnowledgeAS,GoubaultLLR19disc,DitmarschGLLR21}.
In~\cite{gandalf-journal}, we left open the question of a logical obstruction to the solvability of $k$-set agreement, which was later given by Yagi and Nishimura~\cite{yagiNishimura2020TR} using the notion of distributed knowledge~\cite{HalpernM90}, in a sense a higher-dimensional version of knowledge.

\subparagraph{Systems with detectable crashes.}
In this paper, we wish to extend the work of~\cite{gandalf-journal} by lifting the restriction to ``pure'' simplicial complexes.
In distributed computing, pure complexes can be used to analyse the basic \emph{wait-free} shared-memory model of computation~\cite{waitFree91}. 
However, impure\footnote{Throughout this paper, the adjective ``impure'' usually stands for ``not necessarily pure''.} complexes also show up in many situations:
perhaps the most simple one is the \emph{synchronous crash model}, where processes may fail by crashing\footnote{In the distributed computing literature, agents are called \emph{processes}, and when a process stops its execution prematurely, it is said to have \emph{crashed}. In this paper, we will say that agents may \emph{die}.}.
Due to the synchronous nature of the system, when a process crashes, the other processes will eventually know about it. This contrasts with asynchronous systems, where processes can be arbitrarily slow, and there is no way to distinguish a crashed process from a slow one.

Systems where crash-prone processes operate in synchronous rounds have been thoroughly studied since early on in distributed computing, see e.g.~\cite{FischerL82,LynchBook96}. 
At the start of each round, every process sends a message to all the other processes, in unspecified order. A process may crash at any time during the round, in which case only a subset of its messages will be received.
A global clock indicates the end of the round: any message that has not been received by then signifies that the sender has crashed.
Moreover, we usually assume a \emph{full-information} protocol: in each round, the messages sent by the processes consist of its local state at the end of the previous round.  

Figure~\ref{fig-synchEvol} below depicts the simplicial complexes of local states for three processes, after one and two rounds of the synchronous crash model.
In the initial situation (left), the local states are binary input values of the processes, $0$ or $1$.
Each of the $8$ triangles represents a possible global state, i.e.\ an assignment of inputs to processes.
The two other complexes (middle and right) represent the situation after one round and two rounds, respectively.
These complexes are impure: they contain both triangles (representing global states where all three processes are alive) and edges (representing global states where only two agents are alive).
Throughout the paper, we use this model as a running example, starting with~\cref{ex:simpmodel0}.
Further details from the distributed computing perspective can be found in~\cite{HERLIHY20001}. 

Synchronous systems have also been studied using epistemic logic, e.g.\ in the  seminal work of Dwork and Moses~\cite{DworkM90crash}, where a complete characterization of the number of rounds required to reach simultaneous consensus is given, in terms of common knowledge.
The focus however has been on studying solvability of consensus and other problems related to common knowledge, which as mentioned above, depend only on the $1$-dimensional connectivity of epistemic models.

\begin{figure}[h]
    \centering
    \includegraphics[scale=0.6]{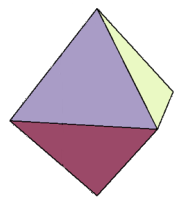}
    \qquad
    \includegraphics[scale=0.6]{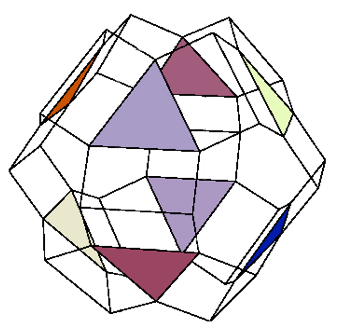}
    \qquad
    \includegraphics[scale=0.6]{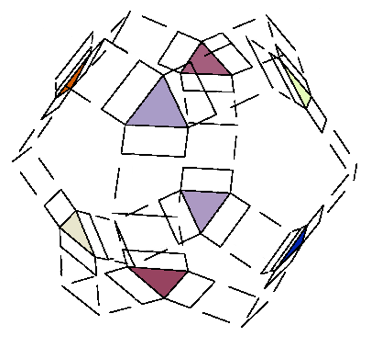}
    \caption{Input complex for three agents starting with binary inputs,  then the complex after one, and after two rounds. At most one agent may die~\cite{HERLIHY20001}.}
    \label{fig-synchEvol}
\end{figure}

\subparagraph{Contributions.}
With the long-term goal of going beyond consensus-like problems, to {$k$-set} agreement, renaming, and other tasks whose solvability depends on  higher dimensional topological connectivity, we introduce in this paper an epistemic logic where agents may die, whose semantics is naturally given by impure simplicial models.


Our approach is guided by the categorical equivalence between $\Sfive$ Kripke models and pure simplicial models, established in~\cite{gandalf-journal}.
It is easy and natural to generalize the class of simplicial models by simply removing the ``pure'' assumption.
However, the main technical challenge resides in finding an equivalent category of Kripke models.
This is achieved in~\cref{sec:frames}, where the categorical equivalence is established in~\cref{thm:equiv} for the frames, and \cref{thm:equiv2} for the models.
Guided by the equivalence with simplicial models, we introduce \emph{partial epistemic models}, whose underlying frame has the following characteristics:
\begin{itemize}
\item Indistinguishability relations must be transitive and symmetric, but may not be reflexive.
\item The frames must be \emph{proper}, in a sense defined in \cref{sec:KPER}.
\end{itemize}
Surprisingly, the morphisms between those frames are also unusual: a world is mapped to a sets of worlds, which must be \emph{saturated} (\cref{def:KPERmor}).

In \cref{sec:models}, we reap the benefits of this equivalence theorem. Modal logics on Kripke models are well understood, and we can then translate results back to simplicial models.
Each of the peculiar conditions that we impose on partial epistemic frames reveals an implicit assumption of simplicial models.

The consequence of losing reflexivity is that the logic is no longer $\Sfive$, but instead $\KBfour$, where the Axiom \textbf{T} does not hold.
This logic is not often considered by logicians; its close cousin $\mathbf{KD45_n}$ being more commonly studied, in order to reason about belief~\cite{Voorbraak92}.
But, as we argue in \cref{sec:aliveanddead,sec:axiomsystem}, $\KBfour$ is an interesting setting to reason about alive and dead agents.
Moreover, the requirement of having proper frames leads us to introduce two additional axioms: the axiom of Non-Emptiness $\mathbf{NE}$ says that at least one agent is alive in every world; and the Single-Agent axioms $\mathbf{SA_a}$ says that if exactly one agent~$a$ is alive, this agent knows everything that is true about the world.
In \cref{sec:completeness}, we claim that the logic $\KBfour$ augmented with these two extra axioms is sound and complete with respect to class of (possibly non-pure) simplicial models.
While soundness is easy to prove, the proof of completeness is more intricate and we leave it for the full version of this work.
Finally in \cref{sec:knowledge-gain}, we prove the so-called \emph{knowledge gain} property, which has been instrumental in applications to impossibility results in distributed computing, see e.g.~\cite{gandalf-journal}.

\subparagraph{Related work.}

A line of work started by Dwork and Moses~\cite{DworkM90crash} studied in great detail the synchronous crash failures model from an epistemic logic perspective.
However, in their approach, the crashed processes are treated the same as the active ones, with a distinguished local state ``\texttt{fail}''.
In that sense, all agents are present in every state, hence they still model the usual epistemic logic $\Sfive$.
Instead of changing the underlying Kripke models as we do here, they introduce new knowledge and common knowledge operators that take into account the non-rigid set of agents (see e.g.~\cite{FHMVbook}, Chapter 6.4).

Giving a formal epistemic semantics to impure simplicial models has also been attempted by van Ditmarsch~\cite{Ditmarsch21}, at the same time and independently from our work.
This approach end up quite different from ours. It describes a two-staged semantics with a \emph{definability relation} prescribing which formulas can be interpreted, on top of which the usual \emph{satisfaction relation} is defined.
This results in a quite peculiar logic: for instance, it does not obey Axiom~\textbf{K}, which is the common ground of all Kripke-style modal logics.
The question of finding a complete axiomatization is left open. 
In contrast, we take a more systematic approach: we first establish a tight categorical correspondence between simplicial models and Kripke models.
Via this correspondence, we translate the standard Kripke-style semantics to simplicial models.
This leads us to the  modal logic $\KBfour$.
We will discuss further the technical differences between our approach and that of~\cite{Ditmarsch21} in~\cref{sec:equiv}.



\section{Background on simplicial complexes and Kripke structures}

\subparagraph{Chromatic simplicial complexes.}
Simplicial complexes are the basic structure of combinatorial topology~\cite{kozlov}.
In the field of fault-tolerant distributed computing~\cite{herlihyetal:2013}, their vertices are usually labelled by process names, often viewed as colours; hence the adjective ``chromatic''.

\begin{definition}
\label{def:simplicial-complex}
A \emph{simplicial complex} is a pair $\C = \langle V,S \rangle$ where $V$ is a set, and $S \subseteq \Pow{V}$ is a family of non-empty subsets of $V$ such that 
for all $v \in V$, $\{v\} \in S$, and
$S$ is downward-closed: for all $X \in S$, if $Y$ is non-empty and $Y\subseteq X$ then $Y \in S$. 

Considering a finite, non-empty set $A$ of \emph{agents}, a \emph{chromatic simplicial complex} coloured by~$A$ is a triple $\langle V,S,\chi \rangle$ where~$\langle V,S \rangle$ is a simplicial complex, and~$\chi : V \to A$ 
assigns colours to vertices such that for every $X \in S$, all vertices of $X$ have distinct colours.
\end{definition}

\label{def:puresimp}
Elements of $V$ are called \emph{vertices}, and are identified with singletons of $S$.
Elements of $S$ are \emph{simplexes}, and the ones that are maximal w.r.t.\ inclusion are \emph{facets}.
The set of facets of~$\C$ is written~$\cF(\C)$.
The \emph{dimension} of a simplex $X \in S$ is $\dim(X) = |X|-1$. 
A simplicial complex $C$ is \emph{pure} if all facets are of the same dimension. 
The condition of having distinct colours for vertices of the same simplex is a fairly strong one: in particular, we will always be allowed to take the (unique) subface of a simplex $X$ of a chromatic simplicial complex with colours in some subset $U$ of $\chi(X)$.

\begin{definition}
\label{def:puresimpmor}
A \emph{chromatic simplicial map} $f : \C \to \D$ from $\C = \langle V,S,\chi \rangle$ to $\D = \langle V',S',\chi' \rangle$ is a function $f : V \to V'$ preserving simplexes, i.e.\ for every $X \in S$, $f(X) \in S'$, and preserving colours, i.e.\ for every $v\in V$, $\chi'(f(v)) = \chi(v)$.
\end{definition}

We denote by $\SimCpx{A}$ the category of chromatic simplicial complexes coloured by~$A$, and 
$\PureSimCpx{A}$ the full sub-category of pure chromatic simplicial complexes on $A$. 

\subparagraph{Equivalence with epistemic frames.}
The traditional possible worlds semantics of (multi-agent) modal logics relies on the notion of Kripke frame.
Let  $A$ be a finite set of agents.

\begin{definition}
\label{def:kripkeframe}
\label{def:kripkeframemor}
A \emph{Kripke frame} $M = \la W, R \ra$ is  a set of \emph{worlds}~$W$, together with an $A$-indexed family of relations on~$W$, $R : A \to \Pow{W \times W}$.
We write $R_a$ rather than $R(a)$, and $u\,R_a\,v$ instead of $(u,v) \in R_a$.
The relation $R_a$ is called the \emph{$a$-accessibility relation}.
%
%
Given two Kripke frames $M=\la W, R \ra$ and $N=\la W',R' \ra$, a \emph{morphism} from  $M$ to $N$ is a function $f : W \to W'$ such that for all $u, v \in W$, for all $a \in A$, 
$u\,{R_a}\,v$ implies $f(u)\,{R'_a}\,f(v)$.
\end{definition}

To model multi-agent epistemic logic $\Sfive$, we additionally require each relation~$R_a$ to be an equivalence relation.
When this is the case, we usually denote the relation by $\sim_a$, and call it the \emph{indistinguishability relation}.
For the equivalence class of~$w$ with respect to $\sim_a$, we write $[w]_{a} \subseteq W$.
Kripke frames satisfying this condition are called \emph{epistemic frames}.
An epistemic frame is \emph{proper} when two distinct worlds can always be distinguished by at least one agent: for all $w,w' \in W$, if $w \neq w'$ then $w \not \sim_a w'$ for some~$a \in A$.
In \cite{gandalf-journal}, we exploited an equivalence of categories between pure chromatic simplicial complexes and proper Kripke frames, to give an interpretation of $\Sfive$ on simplicial models. This allowed us to apply epistemic logics to study distributed tasks.

\begin{theorem}[\cite{gandalf-journal}]
\label{thm:gandalf}
The category of pure chromatic simplicial complexes $\PureSimCpx{A}$ is equivalent to the category of proper epistemic frames $\ProperEFrame{A}$.
\end{theorem}

\begin{example}\label{ex:basicDuality}
The picture below shows an epistemic frame (left) and its associated chromatic simplicial complex (right).
The three agents are named $a,b,c$. 
The three worlds $\{w_1, w_2, w_3\}$ of the epistemic frame correspond to the three facets (triangles) of the simplicial complex.
In the epistemic frame, the $c$-labelled edge between the  worlds $w_2$ and $w_3$ indicates that $w_2 \sim_c w_3$.
Correspondingly, the two facets $w_2$ and $w_3$ of the simplicial complex share a common vertex, labelled by agent~$c$.
Similarly, the worlds $w_1$ and $w_2$ are indistinguishable by both agents $a$ and $b$; so the corresponding facets share their $ab$-labelled edge.

\begin{center}
\begin{tikzpicture}[auto,dot/.style={draw,circle,fill=black,inner sep=1pt,minimum size=4pt},cloud/.style={draw=black,thick,circle,fill=white,inner sep=1pt,minimum size=11pt}]

\node (p) at (-1.5,0) {$w_1$};
\node (q) at (0,0) {$w_2$};
\node (r) at (1.5,0) {$w_3$};
\path (p) edge[bend left] node[above] {$a$} (q)
      (p) edge[bend right] node[below] {$b$} (q)
      (q) edge node[above] {$c$} (r);
 
\node at (3.25,0) {\Large $\cong$};

\draw[thick, draw=black, fill=lipicsLightGray, fill opacity=0.7]
  (5,0) -- (6,-0.577) -- (6,0.577) -- cycle;
\draw[thick, draw=black, fill=lipicsLightGray, fill opacity=0.7]
  (6,-0.577) -- (6,0.577) -- (7,0) -- cycle;
\draw[thick, draw=black, fill=lipicsLightGray, fill opacity=0.7]
  (7,0) -- (8,-0.577) -- (8,0.577) -- cycle;
\node (p') at (5.65,0) {$w_1$};
\node (q') at (6.35,0) {$w_2$};
\node (r') at (7.65,0) {$w_3$};
\node[cloud] (b1) at (5,0) {$c$};
\node[cloud] (g1) at (6,-0.577) {$a$};
\node[cloud] (w1) at (6,0.577) {$b$};
\node[cloud] (b2) at (7,0) {$c$};
\node[cloud] (g2) at (8,-0.577) {$a$};
\node[cloud] (w2) at (8,0.577) {$b$};
\end{tikzpicture}
\end{center}
\end{example}

\section{Partial epistemic frames and simplicial complexes}

\label{sec:frames}
In this section, we generalise~\cref{thm:gandalf} to deal with chromatic simplicial complexes that may not be pure.
For that purpose, we will need to enlarge the class of Kripke frames to be considered, which we call \emph{partial epistemic frames}.
First, we start with our running example of an impure simplicial complex, which has been studied in distributed computing.

\begin{example}[Synchronous crash-failure model, one round, three agents]
\label{ex:simpmodel0}
Consider a set of three processes/agents $A=\{ a,b,c \}$.
For simplicity, we consider a single initial state where the agent $a,b,c$ start with input value $1,2,3$, respectively\footnote{Typically, in distributed computing, many initial assignments of inputs are possible. Thus, we model a situation where the inputs of other processes are not known until a message from them is received.}.
Each agent sends a message to the two other agents (and to itself, for uniformity), containing its input value.
An agent may \emph{crash} during the computation, in which case it stops sending messages. We assume moreover that at most two agents may crash, as in e.g.~\cite{DworkM90crash}. 
At the end of the round, an agent is \emph{alive} if it successfully sent all its messages, and \emph{dead} if it crashed before finishing.
The \emph{view} (or local state) of an alive agent is the set of messages that it received during the round.
Note that an alive agent always sees its own value.
For instance, the four possible views of agent~$a$ after one round are $\{1\}, \{1,2\}, \{1,3\}$ and $\{1,2,3\}$.

This situation is modelled by the chromatic simplicial complex $\C$ on the left of~\cref{fig:broadcast}.
Formally, the vertices of $\C$ are pairs $(a, \view)$ where $a \in A$ and $\view \subseteq \{1,2,3\}$ is its view.
There are $12$ such vertices, $4$ for each agent.
The colouring $\chi(a, \view) = a$ of a vertex is indicated on the picture.
There are $13$ facets $w_0, \ldots, w_{12}$, corresponding to the possible global states at the end of the round.
The middle triangle $w_1=\{ (a,\view_a),(b,\view_b),(c,\view_c)\}$, with $\view_a=\view_b=\view_c=\{1,2,3\}$, represents the execution where no agent dies.
The three isolated vertices, $w_0, w_{11}, w_{12}$ are executions where two agents died. For instance, in $w_0 = \{ (a,\{1\}) \}$, both $b$ and $c$ crashed before sending their value to~$a$.
The $9$ edges represent situations where one agent died, and two survived.
For example, $w_2=\{ (a,\view_a),(c,\view_c)\}$, with $\view_a = \{1,2,3\}$ and $\view_c = \{1,3\}$, represents the execution where $b$ crashed after sending its value to $a$, but not to $c$.
In $w_{10}$, agent $b$ crashed before sending any messages.
\end{example}

\begin{figure}[h]
    \centering
    \begin{tikzpicture}[scale=1.9, auto, font={\small}]
	\begin{pgfonlayer}{nodelayer}
		\node [style=white vertex] (0) at (0, 1) {$a$};
		\node [style=white vertex] (1) at (1, -0.5) {$c$};
		\node [style=white vertex] (2) at (-1, -0.5) {$b$};
		\node [style=white vertex] (3) at (1.5, 2) {$c$};
		\node [style=white vertex] (4) at (2.5, 0.5) {$a$};
		\node [style=white vertex] (5) at (1, -2.25) {$b$};
		\node [style=white vertex] (6) at (-1, -2.25) {$c$};
		\node [style=white vertex] (7) at (-1.5, 2) {$b$};
		\node [style=white vertex] (8) at (-2.5, 0.5) {$a$};
		\node [style=white vertex, label={below:$w_{11}$}] (9) at (-2.5, -1.25) {$b$};
		\node [style=white vertex, label={below:$w_{12}$}] (10) at (2.5, -1.25) {$c$};
		\node [style=white vertex, label={right:$w_0$}] (11) at (0, 3) {$a$};
		\node [style=none] (12) at (-1, -0.5) {};
		\node [style=none] (13) at (0, 1) {};
		\node [style=none] (14) at (1, -0.5) {};
		\node [style=none] (15) at (0, 0) {$w_1$};
	\end{pgfonlayer}
	\begin{pgfonlayer}{edgelayer}
		\draw [style=thick] (2) to node [inner sep=2pt] {$w_5$} (8);
		\draw [style=thick] (8) to node [inner sep=2pt] {$w_4$} (7);
		\draw [style=thick] (7) to node [inner sep=1pt] {$w_3$} (0);
		\draw [style=thick] (0) to node [inner sep=1pt] {$w_2$} (3);
		\draw [style=thick] (3) to node [inner sep=2pt] {$w_{10}$} (4);
		\draw [style=thick] (4) to node [inner sep=2pt] {$w_9$} (1);
		\draw [style=thick] (1) to node [inner sep=2pt] {$w_8$} (5);
		\draw [style=thick] (5) to node [inner sep=2pt] {$w_7$} (6);
		\draw [style=thick] (6) to node [inner sep=2pt] {$w_6$} (2);
		\draw [style=filled path] (13.center)
			 to (14.center)
			 to (12.center)
			 to cycle;
	\end{pgfonlayer}
\end{tikzpicture}
    \hspace{2cm}
    \begin{tikzpicture}[auto, scale=1.4, font={\small}, {every loop/.style}={looseness=2, min distance=4mm}]
	\begin{pgfonlayer}{nodelayer}
		\node [style=empty node] (0) at (0, 0) {$w_1$};
		\node [style=empty node] (1) at (2, 0) {$w_9$};
		\node [style=empty node] (2) at (1, 1.5) {$w_2$};
		\node [style=empty node] (3) at (1, -1.5) {$w_8$};
		\node [style=empty node] (4) at (-1, -1.5) {$w_6$};
		\node [style=empty node] (5) at (-1, 1.5) {$w_3$};
		\node [style=empty node] (6) at (-2, 0) {$w_5$};
		\node [style=empty node] (7) at (-3, 1.5) {$w_4$};
		\node [style=empty node] (8) at (3, 1.5) {$w_{10}$};
		\node [style=empty node] (9) at (0, -3) {$w_7$};
		\node [style=empty node] (10) at (0, 3) {$w_0$};
		\node [style=empty node] (11) at (3, -2) {$w_{12}$};
		\node [style=empty node] (12) at (-3, -2) {$w_{11}$};
	\end{pgfonlayer}
	\begin{pgfonlayer}{edgelayer}
		\draw (0) to node [inner sep=1pt] {\scriptsize $a$} (5);
		\draw (5) to node [inner sep=1pt, swap] {\scriptsize $b$} (7);
		\draw (7) to node [inner sep=1pt, swap] {\scriptsize $a$} (6);
		\draw (6) to node [inner sep=1pt, swap] {\scriptsize $b$} (0);
		\draw (0) to node [inner sep=1pt, swap] {\scriptsize $b$} (4);
		\draw (4) to node [inner sep=1pt, swap] {\scriptsize $c$} (9);
		\draw (9) to node [inner sep=1pt, swap] {\scriptsize $b$} (3);
		\draw (3) to node [inner sep=1pt, swap] {\scriptsize $c$} (0);
		\draw (0) to node [inner sep=1pt, swap] {\scriptsize $c$} (1);
		\draw (1) to node [inner sep=1pt, swap] {\scriptsize $a$} (8);
		\draw (8) to node [inner sep=1pt, swap] {\scriptsize $c$} (2);
		\draw (2) to node [inner sep=1pt] {\scriptsize $a$} (0);
		\draw [in=110, out=70, loop] (10) to node [above] {\scriptsize $a$} ();
		\draw [in=110, out=70, loop] (7) to node [above] {\scriptsize $a,\!b$} ();
		\draw [in=110, out=70, loop] (5) to node [above] {\scriptsize $a,\!b$} ();
		\draw [in=110, out=70, loop] (2) to node [above] {\scriptsize $a,\!c$} ();
		\draw [in=110, out=70, loop] (8) to node [above] {\scriptsize $a,\!c$} ();
		\draw [in=-20, out=20, loop] (1) to node [right] {\scriptsize $a,\!c$} ();
		\draw [in=200, out=160, loop] (6) to node [left] {\scriptsize $a,\!b$} ();
		\draw [in=110, out=70, loop] (12) to node [above] {\scriptsize $b$} ();
		\draw [in=240, out=200, loop] (4) to node [below] {\scriptsize $b,\!c$} ();
		\draw [in=-20, out=-60, loop] (3) to node [below] {\scriptsize $b,\!c$} ();
		\draw [in=110, out=70, loop] (11) to node [above] {\scriptsize $c$} ();
		\draw [in=-110, out=-70, loop] (9) to node [below] {\scriptsize $b,\!c$} ();
		\draw [in=110, out=70, loop] (0) to node [above] {\scriptsize $a,\!b,\!c$} ();
	\end{pgfonlayer}
\end{tikzpicture}
    \caption{A chromatic simplicial complex $\C$ (left), and a proper partial epistemic frame $M$ (right). The three agents are $A = \{a,b,c\}$ and the 13 facets/worlds are labelled $w_0,\ldots,w_{12}$.
    }
    \label{fig:broadcast}
    \label{fig:frame}
\end{figure}
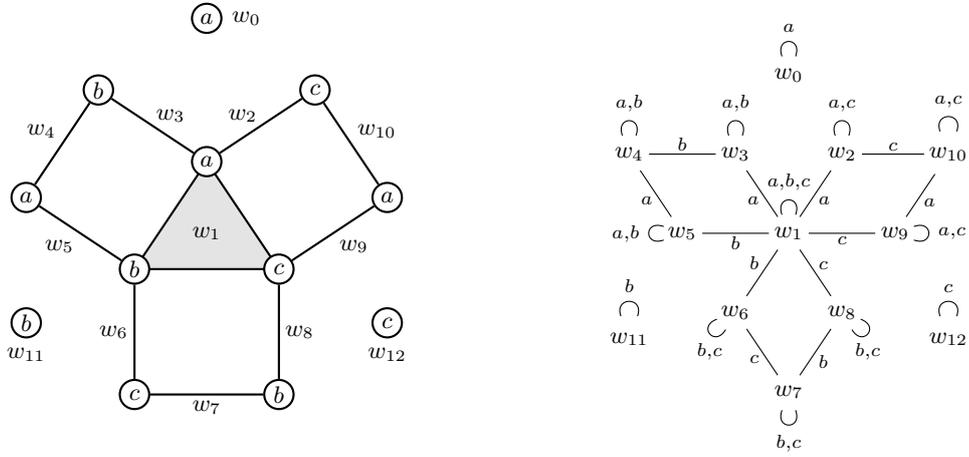

\subsection{Partial epistemic frames}
\label{sec:KPER}

We consider now another type of Krikpe frame, in the spirit of PER semantic models of programming languages and ``Kripke logical partial equivalence relations'' of e.g.~\cite{KPER}.

\begin{definition}
A \emph{Partial Equivalence Relation} (PER) on a set $X$ is a relation ${R \subseteq X \times X}$ which is symmetric and transitive (but not necessarily reflexive).
\end{definition}

The \emph{domain} of a PER $R$ is the set $\dom{R} = \{ x \in X \mid R(x,x)\} \subseteq X$, and it is easy to see that $R$ is an equivalence relation on its domain, and empty outside of it.
Thus, PERs are equivalent to the ``local equivalence relations'' defined in~\cite{Ditmarsch21}.
Recall~$A$ is the set of agents.

\begin{definition}
A \emph{partial epistemic frame} $M = \la W,\sim \ra$ is a Kripke frame such that each relation $(\sim_a)_{a \in A}$ is a PER.
\end{definition}

We say that agent~$a$ is \emph{alive} in a world~$w$ when $w \in \dom{\sim_a}$, i.e., when $w \sim_a w$.
In that case, we write $[w]_a$ for the equivalence class of $w$ with respect to $\sim_a$, within $\dom{\sim_a}$.
We write $\live{w}$ for the set of agents that are alive in world $w$ and $\dead{w}$ for the set of agents that are dead in world $w$ (the complement of $\live{w}$). 
A partial epistemic frame is \emph{proper} if in all worlds, there is at least one agent which is alive, and moreover any two distinct worlds $w,w'$ can be distinguished by at least one agent that is alive in $w$, i.e., 
$
\forall w, w' \in W,\; \exists a \in A,\;\; w\sim_a w \mbox{ and } \left(w \neq w' \implies w \not \sim_a w'\right)
$. 
Note that, by symmetry of~$\neq$, there is also a (possibly different) agent $a'$ that is alive in~$w'$ and can distinguish $w$ and $w'$.

\begin{example} 
\label{ex:proper}
Two partial epistemic frames over the set of agents $A = \{a,b,c\}$ are represented below.
The frame on the left is proper, because agent~$b$ is alive in~$w_1$ and can distinguish between $w_1$ and $w_2$; and agent~$c$ is alive in~$w_2$ and can distinguish between $w_1$ and $w_2$.
The frame on the right is not proper, because there is no agent alive in $w'_2$ that can distinguish between $w'_1$ and $w'_2$.
\begin{center}
\begin{tikzpicture}[auto]
\tikzset{every loop/.style={}}
\node (p) at (-2,0) {$w_1$};
\node (q) at (0,0) {$w_2$};
\path (p) edge[loop above] node {$a,b$} (p)
      (p) edge node {$a$} (q)
      (q) edge[loop above] node {$a,c$} (q);
\end{tikzpicture}
\qquad \qquad
\qquad \qquad
\begin{tikzpicture}[auto]
\tikzset{every loop/.style={}}
\node (p) at (-2,0) {$w'_1$};
\node (q) at (0,0) {$w'_2$};
\path (p) edge[loop above] node {$a,b,c$} (p)
      (p) edge node {$a,b$} (q)
      (q) edge[loop above] node {$a,b$} (q);
\end{tikzpicture}
\end{center}
\end{example}

\begin{example}
\label{ex:frame}
The partial epistemic frame modelling the synchronous crash model of \cref{ex:simpmodel0} is pictured \cref{fig:frame} (right). It has~13 worlds $w_0,\ldots,w_{12}$.
In each world, the set of alive agents can be read off the reflexive ``loop'' edge.
\begin{itemize}
    \item In $w_1$, all three agents $\{a,b,c\}$ are alive. 
    \item In worlds $w_3$, $w_4$ and $w_5$, the two alive agents are $a$ and $b$. In worlds $w_2$, $w_{10}$ and $w_9$, the alive agents are $a$ and $c$. And in worlds $w_6$, $w_7$, $w_8$, agents $b$ and $c$ are alive. 
    \item In $w_0$, only $a$ is alive. In $w_{11}$, only $b$ is alive, and in $w_{12}$, only $c$ is alive. 
\end{itemize}
The accessibility relation is represented by edges labelled with the agents that do not distinguish between the worlds at its extremities. For instance, agent $a$ cannot distinguish between $w_3$ and $w_1$, and agent $b$ cannot distinguish between $w_3$ and $w_4$. It can easily be checked to be a proper partial epistemic frame. 
%
\end{example}

\subparagraph{Morphisms of partial epistemic frames.}
Our notion of morphism for partial epistemic frames differs from the one for a general Kripke frame (\cref{def:kripkeframemor}).
Here again, our definitions are guided by our goal (\cref{thm:equiv}), the equivalence between simplicial maps and morphisms of partial epistemic frames.
\cref{ex:morphisms} below should help motivate our definitions.
The novelty arises when we want a morphism $f$ that maps a world $w$, in which some agents $\live{w}$ are alive, to a world~$w'_1$ where strictly more agents are alive.
In this case, there might exist some other world~$w'_2$,
such that $w'_1 \sim_a w'_2$ for all $a \in \live{w}$.
We claim that such a world $w'_2$ should also be in the image of $w$ by the morphism $f$.
Thus, $f(w)$ is not a world but a set of worlds, which we require to be \emph{saturated}, in the following sense.

\begin{example}
\label{ex:morphisms}
The two pictures below show a chromatic simplicial map $g$ (left) and a morphism $f$ of partial epistemic frames (right).
The simplicial map $g$ is uniquely specified by the preservation of colours: it maps the edge $w_0$ onto the vertical $ab$-coloured edge of the complex on the right.
The morphism $f$ is defined by $f(w_0) = \{ w'_1, w'_2 \}$.
We will see in \cref{sec:equiv} how to relate these morphisms: one can be built from the other, and vice-versa.
\begin{center}
\begin{tikzpicture}[auto,scale=1.1,cloud/.style={draw=black,thick,circle,fill=white,inner sep=1pt,minimum size=11pt}]
\draw[thick, draw=black, fill=lipicsLightGray, fill opacity=0.7]
  (5,0) -- (6,-0.577) -- (6,0.577) -- cycle;
\draw[thick, draw=black, fill=lipicsLightGray, fill opacity=0.7]
  (6,-0.577) -- (6,0.577) -- (7,0) -- cycle;
\node (p') at (5.65,0) {$w'_1$};
\node (q') at (6.35,0) {$w'_2$};
\node[cloud] (b1) at (5,0) {$c$};
\node[cloud] (g1) at (6,-0.577) {$b$};
\node[cloud] (w1) at (6,0.577) {$a$};
\node[cloud] (b2) at (7,0) {$c$};
\node[cloud] (g0) at (3,-0.577) {$b$};
\node[cloud] (w0) at (3,0.577) {$a$};
\draw[thick] (g0) edge node {$w_0$} (w0);
\draw[->] (3.6,0) -- (4.5,0) node[midway] {$g$};
\end{tikzpicture}
\qquad \qquad
\qquad
\begin{tikzpicture}[auto]
\tikzset{every loop/.style={}}
\node (p0) at (-5,0) {$w_0$};
\path (p0) edge[loop above] node {$a,b$} (p0);
\node (p) at (-2,0) {$w'_1$};
\node (q) at (0,0) {$w'_2$};
\path (p) edge[loop above] node {$a,b,c$} (p)
      (p) edge node {$a,b$} (q)
      (q) edge[loop above] node {$a,b,c$} (q);
\draw[->] (-4,0.3) -- (-3,0.3) node[midway] {$f$};
\end{tikzpicture}
\end{center}
\end{example}

\begin{definition}
Given a partial epistemic frame $M = \la W, \sim \ra$, a subset of agents $U \subseteq A$, and a world $w \in W$, let
$
  \sat_{U}(w) = \{ w' \in W \mid w \sim_a w' \mbox{ for all } a \in U \}.
$
\end{definition}
The saturation requirement will be crucial in~\cref{sec:equiv} when we establish the equivalence of categories between partial epistemic frames and chromatic simplicial complexes.

\begin{definition}
\label{def:KPERmor}
Let $M=\langle W,\sim\rangle$ and $N=\langle W', \sim'\rangle$ be two partial epistemic frames.
A morphism of partial epistemic frame from $M$ to $N$ is a function $f: W \rightarrow \Pow{W'}$ such that
\begin{itemize}
\item \emph{(Preservation of $\sim$)} for all $a \in A$, for all $u,v \in W$,  $u \sim_a v$ implies $u' \sim'_a v'$, for all $u' \in f(u)$ and  $v' \in f(v)$,
\item \emph{(Saturation)}
for all $u \in W$, there exists $u' \in f(u)$ such that $f(u) = \sat_{\live{u}}(u')$. 
\end{itemize}
Composition of morphisms is defined by $(g\circ f)(u) = \sat_{\live{u}}(w)$, for some $v \in f(u)$ and $w \in g(v)$.
\end{definition}

Let us check that the composite $g \circ f$ above is well-defined, i.e., that it does not depend on the choice of $v \in f(u)$ and $w \in g(v)$.
Assume we pick $v' \in f(u)$ and $w' \in g(v')$ instead.
Then $v \sim_a v'$ for all $a \in \live{u}$, because $f(u)$ is saturated. And by preservation of $\sim$, we get $w \sim_a w'$ for all $a \in \live{u}$, that is, $\sat_{\live{u}}(w) = \sat_{\live{u}}(w')$.

The first condition of a morphism $f$ of partial epistemic frame  above means that worlds that are indistinguishable by some agent~$a$ should have images composed of worlds that are indistinguishable by~$a$. 
The second condition states that the image of a world $u$ of $M$ is ``generated'' by a world $u'$ of $N$, as the set of all worlds of $N$ that cannot be distinguished from $u'$ by the agents alive in $u$.
In particular, notice that the saturation condition implies that $f(u)$ is always non-empty.

The next proposition says that, on proper frames, the only case when $f(u)$ can be multivalued is when $\live{u} \subsetneq \live{u'}$ for every $u'$ in $f(u)$.

%

\begin{proposition}
\label{prop:singleton}
Let $M=\langle W,\sim\rangle$ and $N=\langle W', \sim'\rangle$ be two partial epistemic frames, and
 $f : M \to N$ be a morphism.
For all $u \in W$ and $u' \in f(u)$, $\live{u} \subseteq \live{u'}$.
Moreover, if $N$ is proper and $\live{u} = \live{u'}$, then $f(u)=\{u'\}$.
\end{proposition}

\begin{proof}
The first fact is a direct consequence of the preservation of~$\sim$.
For the second one, let $u' \in f(u)$ such that $\live{u'}=\live{u}$.
Assume by contradiction that there is $u'' \in f(u)$ with $u'' \neq u'$.
By saturation, we have $u''\sim_a u'$ for all $a \in \live{u}=\live{u'}$.
This is impossible since $N$ is proper. 
\end{proof}


The category of partial epistemic frames with set of agents $A$ is denoted by $\PEFrame{A}$,
and the full subcategory of proper partial epistemic frames is denoted by $\ProperPEFrame{A}$. 
Note that the category of proper epistemic frames $\ProperEFrame{A}$ 
is a full subcategory of~$\ProperPEFrame{A}$.
Indeed, in an epistemic frame all agents are alive in all worlds, so by \cref{prop:singleton} morphisms between proper epistemic frames are single-valued.
Then \cref{def:KPERmor} reduces to the standard notion of Kripke frame morphisms (\cref{def:kripkeframemor}).

\subsection{Equivalence between chromatic simplicial complexes and partial epistemic frames}
\label{sec:equiv}

In this section, we show how to canonically associate a proper partial epistemic frame with any chromatic simplicial complex, and vice-versa.
In fact, we have an equivalence of categories, meaning this correspondence can be extended to morphisms too (see \cref{ex:morphisms}).
We construct functors $\kappa : \SimCpx{A} \rightarrow \ProperPEFrame{A}$ and ${\sigma : \ProperPEFrame{A} \rightarrow \SimCpx{A}}$ and show that they form an equivalence of categories in \cref{thm:equiv}.
%
A similar correspondence appears in~\cite{Ditmarsch21}, with two differences:
\begin{itemize}
\item They only show the equivalence between the objets of those categories, while we also deal with morphisms.
To achieve this, we had to define morphisms of partial epistemic frames (\cref{def:KPERmor}), since the standard notion does not work.
\item They only show that $\kappa \circ \sigma(M)$ is bisimilar to $M$, while we prove a stronger result, that there is an isomorphism.
To achieve this, we had to impose the condition of $M$ being proper, which is not considered in~\cite{Ditmarsch21}.
\end{itemize}

\begin{definition}[Functor $\kappa$]
\label{def:adjF}
Let $\C = \la V, S, \chi \ra$ be a chromatic simplicial complex on the set of agents~$A$. Its associated partial epistemic frame is $\kappa(\C)=\la W, \sim \ra$, where $W := \cF(\C)$ is the set of facets of $\C$, and the PER $\sim_a$ 
is given by $X \sim_a Y$ if $a \in \chi(X \cap Y)$ (for $X,Y \in \cF(\C)$).

The image of a morphism $f: \C \rightarrow \D$ in $\SimCpx{A}$, is the morphism $\kappa(f) : \kappa(\C) \to \kappa(\D)$ in $\ProperPEFrame{A}$ that takes a facet $X \in \cF(\C)$ to 
$\kappa(f)(X) \;=\; \{Z \in \cF(\D) \mid f(X) \subseteq Z \}$.
\end{definition}

\begin{example}
In \cref{fig:broadcast}, the simplicial complex $\C$ on the left 
is mapped by $\kappa$ to the partial epistemic frame $M = \kappa(\C)$ on the right. 
The epistemic frame $M$ contains a world per facet $w_0,\ldots,w_{12}$ of the simplicial complex. The reflexive ``loops'' in the $M$, indicating which agents are alive in a given world, are labelled with the colours of the corresponding facet.
For instance, $w_1 \sim_{\{a,b,c\}} w_1$ but $w_3 \sim_{\{a,b\}} w_3$ only; because $w_3$ in $\C$ is an edge whose extremities have colours $a$ and $b$.

The action of $\kappa$ on morphisms can be seen in \cref{ex:morphisms}, which depicts a simplicial map~$g$ and its associated morphism of partial epistemic frames, $f = \kappa(g)$.
\end{example}

\noindent
We now check that $\kappa$ is a well-defined functor from $\SimCpx{A}$ to $\ProperPEFrame{A}$.

\begin{proposition}
\label{lemma:proper}
$\kappa(\C)$ is a proper partial epistemic frame.
\end{proposition}
\begin{proof}
The relation $\sim_a$ on facets is easily seen to be symmetric and transitive, because there can be at most one vertex $v \in X \cap Y$ with $\chi(v) = a$.
To show that $\kappa(\C)$ is proper, consider two worlds $X$ and $Y$ in $\kappa(\C)$, i.e., two facets of $\C$.
In simplicial complexes, $X \neq Y$ implies that at least one vertex of $X$, say $v$, does not belong to $Y$: otherwise, we would have $X \subseteq Y$ so $X$ would not be a facet.
Let $a = \chi(v)$ be the colour of~$v$.
Then $a$ is alive in~$X$ because $a \in \chi(X \cap X)$;
and $X \not \sim_a Y$ because $v \not \in X \cap Y$ and there can be only one vertex with colour $a$ in $X$.
\end{proof}

\begin{proposition}
\label{prop:mor}
$\kappa(f)$ is a morphism of partial epistemic frames from $\kappa(\C)$ to $\kappa(\D)$.
\end{proposition}
\begin{proof}
Assume $X$ and $Y$ are facets of $\C = \la V,S,\chi \ra$ such 
that $X \sim_a Y$ in $\kappa(\C)$. 
So there is a vertex $v \in V$ such that $v \in X \cap Y$ and $\chi(v) = a$.
Therefore $f(v)$ is in all facets $Z \in \kappa(\D)$ such that $f(X)\subseteq Z$ and all facets $T \in \kappa(\D)$ such that $f(Y) \subseteq T$. As $\chi(f(v))=a$, this means that $a \in \chi(Z \cap T)$, hence, for all $Z\in \kappa(f)(X)$ and $T \in \kappa(f)(Y)$, $Z \sim_a T$.
Furthermore, $\kappa(f)(X)$ as defined is obviously saturated, so $\kappa(f)$ is a morphism of partial epistemic frames.
\end{proof}

\begin{proposition}
\label{prop:kappafunctorial}
$\kappa$ is functorial, i.e.\ $\kappa(g \circ f) = \kappa(g) \circ \kappa(f)$.
\end{proposition}
\begin{proof}
Let $f : \C \to \D$ and $g : \D \to \E$ be two chromatic simplicial maps.
By definition, for a world/facet $X \in \kappa(\C)$, we have $\kappa(g \circ f)(X) = \{Z' \in \cF(\E) \mid (g \circ f)(X) \subseteq Z' \}$,
while $(\kappa(g) \circ \kappa(f))(X) = \sat_{\chi(X)}(Z)$ for some facets $Z \in \kappa(g)(Y)$ and $Y \in \kappa(f)(X)$.
We show that they are equal.

Consider $Z'$ such that $(g \circ f)(X) \subseteq Z'$; we need to show that $Z' \sim_a Z$ for all $a \in \chi(X)$. Indeed, let $v$ be the $a$-coloured vertex of $X$.
Then $(g \circ f)(v) \in Z'$ by assumption, and $(g \circ f)(v) \in Z$ because $f(v) \in Y$. 
So there is an $a$-coloured vertex $(g \circ f)(v) \in Z' \cap Z$.

Conversely, let $Z' \in \sat_{\chi(X)}(Z)$, i.e.\ $Z' \sim_a Z$ for all $a \in \chi(X)$.
Let $v$ be a vertex of~$X$, and let $a = \chi(v)$.
Since $f(v) \in Y$, we have $(g \circ f)(v) \in Z$.
Since $Z$ can have only one $a$-colored vertex and $a \in \chi(Z' \cap Z)$, we get $(g \circ f)(v) \in Z'$. Thus $(g \circ f)(X) \subseteq Z'$ as required.
\end{proof}

Conversely, we now consider a partial epistemic frame $M=\la W,\sim \ra$ on the set of agents~$A$, 
and we define the associated chromatic simplicial complex $\sigma(M)$.
Intuitively, each world $w \in W$ where $k+1$ agents are alive will be represented by a facet $X_w$ of dimension~$k$, whose vertices are coloured by $\live{w}$.
Such facets must then be ``glued'' together according to the indistinguishability relations.
Formally, this is done by the following quotient construction:


\begin{definition}[Functor $\sigma$ on objects]
\label{def:adjG}
Let $M = \la W, \sim \ra$ be a partial epistemic frame.
Its associated chromatic simplicial complex is $\sigma(M) = \la V, S, \chi \ra$, where:
\begin{itemize}
\item The set of vertices is $V = \{(a,[w]_a) \mid w \in W, a \in \live{w} \}$. We denote such a vertex $(a,[w]_a)$ by $v^w_{a}$ for succinctness; but note that $v^w_{a} = v^{w'}_{a}$ when $w \sim_a w'$.
\item The facets are of the form $X_w = \{ v^w_{a} \mid a \in \live{w}\}$ for each $w \in W$; and the set $S$ consists of all their sub-simplexes.
\item The colouring is given by $\chi(v^w_a) = a$.
\end{itemize}
\end{definition}

It is straightforward to see that this is a chromatic simplicial complex.
We now check that there is indeed one distinct facet of $\sigma(M)$ for each world of $M$.

\begin{lemma}
\label{lemma:facets}
If $M$ is proper, the facets of $\sigma(M)$ are in bijection with the worlds of~$M$. 
\end{lemma}
\begin{proof}
Each world $w \in W$ is associated with the simplex $X_w = \{ v^w_{a} \mid a \in \live{w}\}$.
We need to prove that these simplexes are indeed facets, and that they are distinct for $w \neq w'$.
It suffices to show that for all $w \neq w'$, $X_w \not \subseteq X_{w'}$. 
Since $M$ is proper, there exists an agent~$a$ which is alive in $w$ such that $w \not \sim_a w'$.
Then, either $a$ is alive in $w'$, in which case $v^w_a \neq v^{w'}_a$, or $a$ is dead in $w'$.
In both cases, $v^w_a$ is not a vertex of $X_{w'}$ so $X_w \not \subseteq X_{w'}$.
\end{proof}

\begin{example}
In \cref{fig:broadcast}, the partial epistemic frame $M$ on the right is mapped by $\sigma$ onto the simplicial complex $\C = \sigma(M)$ on the left. 
Each world $w_0,\ldots,w_{12}$ of $M$ is turned into a facet of the simplicial complex $\sigma(M)$, whose dimension is the number of alive agents minus one.
These facets are glued along the sub-simplexes whose colours are the agents that cannot distinguish between two worlds.
For instance, world $w_1$ is associated with the facet of the same name, with 3 colours, hence of dimension 2 (the central triangle). On the other hand, the world $w_3$ turns into an edge (dimension 1), glued to the triangle $w_1$ along the vertex with colour $a$, because $w_1 \sim_a w_3$.
\end{example}

%
%
%

We also define the action of $\sigma$ on morphisms of partial epistemic frames: 

\begin{definition}[Functor $\sigma$ on morphisms]
\label{def:Gmor}
Now let $f : M \to N$ be a morphism in $\ProperPEFrame{A}$.
We define the simplicial map $\sigma(f) : \sigma(M) \to \sigma(N)$ as follows. For each vertex of $\sigma(M)$ of the form $v^w_a$ with $w \in W$, we pick any $w' \in f(w)$ and define $\sigma(f)(v^w_a) = v^{w'}_a$.
\end{definition}

To check that this is well-defined, we need to show that the simplicial map $\sigma(f)$ does not depend on the choices of $w$ and $w'$.
Assume we pick a different world $u' \in f(w)$, $u' \neq w'$.
By the saturation property of~$f$ we have $u' \sim'_a w'$, so $v^{u'}_a = v^{w'}_a$. Hence $\sigma(f)(v^w_a)$ is a uniquely defined vertex of $\sigma(N)$.
Now, assume that the vertex $v^w_a$ of $\sigma(M)$ could also be described as~$v^u_a$ with $u \in W$.
Since $v^w_a = v^u_a$, we have $w \sim_a u$ in $M$.
By the preservation property of~$f$, for every $u' \in f(u)$ we have $u' \sim'_a w'$, so $v^{u'}_a = v^{w'}_a$.
Once again, the choice of $w \in W$ does not influence the definition of $\sigma(f)$. 

It is easy to check that $\sigma(f)$ is indeed a chromatic simplicial map: preservation of colours is obvious by construction; and for the preservation of simplexes, notice that each facet $X_w$ of $\sigma(M)$ is mapped into the facet $X_{w'}$ of $\sigma(N)$, for some $w' \in f(w)$.
However, note that $\sigma(f)(X_w)$ might not in general be a facet; we only know that $\sigma(f)(X_w) \subseteq X_{w'}$.

\begin{proposition}
\label{prop:sigmafunctorial}
$\sigma$ is functorial, i.e.\ $\sigma(g \circ f) = \sigma(g) \circ \sigma(f)$.
\end{proposition}
\begin{proof}
Let $f : M \to N$ and $g : N \to P$ be morphisms of partial epistemic frames.
Let $v^w_a$ be a vertex of $\sigma(M)$, where $w \in W$ is a world of~$M$.
By definition, $\sigma(g \circ f)(v^w_a) = v^{w''}_a$ where $w'' \in (g \circ f)(w)$; whereas $(\sigma(g) \circ \sigma(f))(v^w_a) = v^{y''}_a$ where $y'' \in g(y')$ and $y' \in f(w)$.
To show that they are the same vertex, we need to prove that $w'' \sim_a y''$.
By definition of $(g \circ f)(w)$, there exists $x' \in f(w)$ and $x'' \in g(x')$ such that $w'' \sim_a x''$.
Since $w \sim_a w$, we have $x' \sim_a y'$ by the preservation property of~$f$, and then $x'' \sim_a y''$ again by preservation.
Finally,  $w'' \sim_a y''$ by transitivity.
\end{proof}

Now we can state the main technical result of this paper: 

\begin{theorem}
\label{thm:equiv}
$\kappa$ and $\sigma$ define an equivalence of categories between $\ProperPEFrame{A}$ and $\SimCpx{A}$.
\end{theorem}
\begin{proof}
We have already seen that $\kappa$ and $\sigma$ are well-defined functors, it remains to show that:
\begin{enumerate}[(i)]
\item \label{item1} The composite $\kappa \circ \sigma$ is naturally isomorphic to the identity functor on $\ProperPEFrame{A}$.
\item The composite $\sigma \circ \kappa$ is naturally isomorphic to the identity functor on $\SimCpx{A}$.
\end{enumerate}

{\bf (i)} Consider a partial epistemic frame $M=\la W,\sim \ra$ in $\ProperPEFrame{A}$. 
By definition, $\kappa\sigma(M) = \la F,\sim' \ra$ where $F$ is the set of facets of $\sigma(M)$.
By \cref{lemma:facets} there is a bijection $W \cong F$, where a world $w \in W$ if associated with the facet $X_w = \{ v^w_a \mid a \in \live{w} \}$ of $\sigma(M)$.
Furthermore, for all $w,w' \in W$, $w \sim_a w'$ iff $X_w \sim'_a X_{w'}$.
Indeed, $w \sim_a w' \iff v^w_a = v^{w'}_a \iff a \in \chi(X_w \cap X_{w'})$.
Hence, $\kappa\sigma(M)$ and $M$ are isomorphic partial epistemic frames. 

Consider a morphism of partial epistemic frames $f: M \rightarrow N$, with $M=\langle W, \sim\rangle$ and $N=\langle W',\sim\rangle$. 
By definition, $\kappa\sigma(f)$ takes a facet $X_w$ of $\sigma(M)$ to a set of facets of $\sigma(N)$, $\kappa\sigma(f)(X_w)=\{Z \in \sigma(N) \ \mid \ \sigma(f)(X_w) \subseteq Z \}$.
We want to show that this set is equal to $\{X_{w'} \mid w' \in f(w)\}$.
Let $w' \in f(w)$. By definition, $\sigma(f)$ maps each vertex $v^w_a$ of $X_w$ to $v^{w'}_a$, so $\sigma(f)(X_w) \subseteq X_{w'}$.
Conversely, assume $\sigma(f)(X_w) \subseteq Z$. Since $Z$ is a facet of $\sigma(N)$, $Z = X_{w'}$ for some $w' \in W'$.
For each $a \in \live{w}$, the vertex $v^w_a$ of $X_w$ is mapped by $\sigma(f)$ to~$v^{x'}_a$, for $x' \in f(w)$. 
But since $\sigma(f)(v^w_a) \in Z$, we must have $v^{x'}_a = v^{w'}_a$, so $x' \sim_a w'$.
By the saturation property of~$f$, $x' \in f(w)$ implies $w' \in f(w)$ as required.
Therefore $\kappa\sigma$ is an isomorphism also on morphisms of partial epistemic frames.

{\bf (ii)} Consider now a chromatic simplicial complex $\C = \la V, S, \chi \ra$.
Then $\sigma\kappa(\C) = \la V',S',\chi' \ra$ has vertices of the form $V' = \{ v^{Z}_a \mid Z \in \cF(\C) \mbox{ and } a \in \chi(Z) \}$.
We must exhibit a bijection $V \cong V'$ which is a chromatic simplicial map in both directions.
Given $u \in V$ of colour $a$, we map it to $v^Z_a$ where~$Z$ is any facet of~$\C$ that contains~$u$.
This is well-defined since any other facet~$Z'$ also containing~$u$ gives rise to the same vertex $v^{Z'}_a = v^Z_a$, because $Z' \sim_a Z$ in~$\kappa(\C)$.
This map is obviously chromatic, and preserves simplexes because any simplex $Y \in S$ contained in a facet $Z \in \cF(\C)$ will be mapped to $\{ v^Z_a \mid a \in \chi(Y) \} \subseteq X_Z \in \cF(\sigma\kappa(\C))$.
Conversely, we map a vertex $v^{Z}_a \in V'$ to the $a$-coloured vertex of~$Z$.
This is also chromatic, and preserves simplexes because any sub-simplex of $X_Z$ is mapped to a sub-simplex of $Z$.
It is easy to check that our two maps form a bijection, 
therefore $\C$ and $\sigma\kappa(\C)$ are isomorphic.

Lastly, consider a chromatic simplicial map $f: \C \to \D$ with $\C = \langle V,S,\chi\rangle$ and $\D = \langle U,R,\zeta\rangle$.
As above, we write $V'$ and $U'$ for the vertices of $\sigma\kappa(\C)$ and $\sigma\kappa(\D)$, respectively.
By definition, $\sigma\kappa(f)$ maps a vertex $v^{Z}_a \in V'$, with $Z \in \cF(C)$, to the vertex $v^{Y}_a \in U'$, with $Y \in \kappa(f)(Z)$.
So by definition of $\kappa(f)$, $f(Z) \subseteq Y$.
To prove that $\sigma\kappa(f)$ agrees with $f$ up to the isomorphism of the previous paragraph, we need to show that $f$ sends the $a$-coloured vertex of~$Z$ to the $a$-coloured vertex of~$Y$. But this is immediate since $f(Z) \subseteq Y$ and $f$ is chromatic.
\end{proof}

\begin{remark}
Note that the equivalence of categories of \cref{thm:equiv} strictly extends the one of \cite{gandalf-journal}, which was restricted to pure chromatic simplicial complexes on one side and proper epistemic frames on the other.
Indeed, if $\C$ is a pure simplicial complex of dimension~$|A|-1$, it is easy to check that $\kappa(\C)$ is an epistemic frame, since all agents are alive in all worlds.
Moreover, by \cref{prop:singleton}, the morphisms between those frames are single-valued; so we recover the usual notion of Kripke frame morphism that we had in~\cite{gandalf-journal}.
Similarly, when $M$ is a proper epistemic frame, the associated simplicial complex $\sigma(M)$ is pure of dimension $|A|-1$.
When restricted to these subcategories, $\sigma$ and $\kappa$ are the same functors as in \cite{gandalf-journal}.
\end{remark}

\section{Epistemic logics and their simplicial semantics}
\label{sec:models}



Let $\AP$ be a countable set of atomic propositions and $A$ a finite set of agents.
The syntax of epistemic logic formulas $\phi \in \mathcal{L}_K$ is generated by the following BNF
grammar:
$$
\varphi ::= p \mid \neg\varphi \mid \varphi \land \varphi \mid
K_a\varphi \qquad p \in \AP,\ a \in A
$$
We will also use the derived operators, defined as usual:
$
\phi \lor \psi := \neg(\neg \phi \land \neg \psi), \ 
{\phi \Rightarrow \psi} := \neg \phi \lor \psi, \
\true := p \lor \neg p, \
\false := \neg \true
$.
Moreover, we assume that the set of atomic propositions is split into a disjoint union of sets, indexed by the agents: $\AP = \bigcup_{a \in A} \AP_a$.
This is usually the case in distributed computing where the atomic propositions represent the local state of a particular agent~$a$.
For $U \subseteq A$, we write $\AP_U := \bigcup_{a \in U} \AP_a$ for the set of atomic propositions concerning the agents in~$U$.

\subsection{Partial epistemic models and Simplicial models}

In \cref{sec:frames}, we exhibited the equivalence between \emph{partial epistemic frames} and \emph{chromatic simplicial complexes}.
In order to give a semantics to epistemic logic, we need to add some extra information on those structures, by labelling the worlds (resp., the facets) with the set of atomic propositions that are true in this world.
This gives rise to the notions of \emph{partial epistemic models} and \emph{simplicial models}, respectively.
As we shall see, the equivalence of \cref{thm:equiv} extends to models in a straightforward manner.




\begin{definition}
\label{def:partial-models}
A \emph{partial epistemic model}  $\model = \la W,\sim,L \ra$  over the set of agents $A$ consists of a partial epistemic frame 
 $\la W, \sim \ra$ on~$A$, together with function $L : W \to \Pow{\AP}$.

Given another partial epistemic model $\model' = \la W',\sim',L' \ra$, a \emph{morphism} of partial epistemic models $f : \model \to \model'$
is a morphism of the underlying partial epistemic frames such that for every world $w \in W$ and $w' \in f(w)$,
$L'(w') \cap \AP_{\live{w}} =  L(w) \cap \AP_{\live{w}}$.
\end{definition}

Let us give some intuition about Definition~\ref{def:partial-models}. The set $L(w)$ contains the atomic propositions that are true in the world $w$.
Note that partial epistemic models are simply Kripke models (in the usual sense of modal logics), such that all the accessibility relations $(\sim_a)_{a \in A}$ are PERs.
In particular, one might have expected the additional restriction $L(w) \subseteq \AP_{\live{w}}$, saying that a world only contains atomic propositions concerning the alive agents.
As we will see in Example~\ref{ex:simpmodel}, there are practical cases where this is not desirable, so we do not impose this.
Secondly, recall from Definition~\ref{def:KPERmor} that, given a morphism $f$ of partial epistemic frames, a world $w \in W$ and a world $w' \in f(w)$, it is possible that $w'$ has strictly more alive agents than $w$.
When that is the case, in the definition of model morphisms above, we require that the labellings $L$ and $L'$ are preserved only for those agents that are alive in $w$.

A partial epistemic model is called \emph{proper} when the underlying frame is proper in the sense of \cref{sec:KPER}.
A \emph{pointed partial epistemic model} is a pair $(\model,w)$ where $w$ is a world of $\model$. A \emph{morphism} of pointed partial epistemic models $f : (\model,w) \rightarrow (\model',w')$ is a morphism of the partial epistemic models $f : \model \to \model'$ that preserves the distinguished world, i.e.\ $w' \in f(w)$. 
We denote by $\PMod{A}{\AP}$ (resp.\ $\PModStar{A}{\AP}$) the category of (resp.\ pointed) proper partial epistemic models over the set of agents $A$ and atomic propositions $\AP$.

\medskip
Recall from \cref{thm:equiv} that the worlds of a partial epistemic frame correspond to the facets of the associated chromatic simplicial complex.
Thus, to get a corresponding notion of simplicial model, we label the facets by sets of atomic propositions:

\begin{definition}
A \emph{simplicial model} $\C = \la V, S, \chi, \ell \ra$ over the set of agents~$A$ consists of a chromatic simplicial complex $\la V,S,\chi \ra$ together with a labelling $\ell : {\cal F}(\C) \to \Pow{\AP}$
that associates with each facet $X \in {\cal F}(\C)$ a set of atomic propositions.

Given another simplicial model $\D = \la V', S', \chi', \ell' \ra$, a \emph{morphism} of simplicial models $f : \C \to \D$ is a chromatic simplicial map such that
for all $X \in \cF(\C)$ and all $Y \in \cF(\D)$, if $f(X) \subseteq Y$  then $\ell'(Y) \inter \AP_{\chi(X)} = \ell(X) \inter \AP_{\chi(X)}$.
%
\end{definition}



A \emph{pointed simplicial model} is a pair $(\C,X)$ where $\C$ is a simplicial model and $X$ is a facet of~$\C$. 
A \emph{morphism} $f: (\C,X) \rightarrow (\D,Y)$ of pointed simplicial models is a morphism  $f : \C  \to \D$  such that $f(X) \subseteq Y$. 
We denote by $\SMod{A}{\AP}$ (resp. $\SModStar{A}{\AP}$) the category of (resp. pointed) simplicial models over the set of 
agents $A$ and atomic propositions $\AP$.
The equivalence of \cref{thm:equiv} can be extended to models and pointed models:
 
 \begin{theorem}
 \label{thm:equiv2}
$\kappa$ and $\sigma$ induce an equivalence of categories between $\SMod{A}{\AP}$ (resp.\ $\SModStar{A}{\AP}$) and $\PMod{A}{\AP}$
(resp.\ $\PModStar{A}{\AP}$).
 \end{theorem}
\begin{proof}
For a simplicial model $\C = \la V,S,\chi,\ell \ra$, recall that the worlds of the associated partial epistemic frame are the facets of $\C$; so the labelling in $\kappa(\C)$ is $L(X) = \ell(X)$ for $X \in \cF(\C)$.
For a partial epistemic model $\model = \la W, \sim, L \ra$, recall that the facets of the associated chromatic simplicial complex are of the form $X_w$ for $w \in W$; so to define $\sigma(\model)$, we set $\ell(X_w) = L(w)$.
For the pointed version, we similarly define $\kappa(\C,X) = (\kappa(\C),X)$ and $\sigma(\model,w) = (\sigma(\model), X_w)$.

Checking that this is indeed an equivalence of category is an immediate consequence of \cref{thm:equiv}.
The only detail to check is that the extra conditions on morphisms are preserved:
if $f$ is a morphism of (pointed) simplicial models, then $\kappa(f)$ is a morphism of (pointed) partial epistemic models.
Indeed, $f(X) \subseteq Y$ implies that $Y \in \kappa(f)(X)$ by definition of $\kappa(f)$.
Similarly, if $g$ is a morphism of (pointed) partial epistemic models, then $\sigma(g)$ is a morphism of (pointed) simplicial models.
\end{proof}


\begin{example}
\label{ex:simpmodel}
In distributed computing, we are usually interested in reasoning about the input values of the various agents, 
so the set of atoms is $\AP = \{ \inputprop{a}{x} \mid a \in A, x \in \Values \}$.
The meaning of the atomic proposition $\inputprop{a}{x}$ is that ``agent~$a$ has input value~$x$''.

Consider again the chromatic simplicial complex $\C$ of \cref{ex:simpmodel0}.
Here, we have three agents $A = \{a,b,c\}$ and three values $\Values = \{1,2,3\}$.
Hence, we can construct a simplicial model via the following labelling of facets~$\ell : \cF(\C) \to \Pow{\AP}$.
\begin{itemize}
\item For the middle triangle $w_1$, all three agents are alive and successfully communicated their input values. So, it makes sense to set $\ell(w_1) = \{\inputprop{a}{1}, \inputprop{b}{2}, \inputprop{c}{3}\}$.
\item Perhaps more surprisingly, we also choose the same labelling for the six edges adjacent to~$w_1$: $\ell(w_2) = \ell(w_3) = \ell(w_5) = \ell(w_6) = \ell(w_8) = \ell(w_9) =  \{\inputprop{a}{1}, \inputprop{b}{2}, \inputprop{c}{3}\}$.
Indeed, consider for instance the world $w_2$, where agent~$b$ crashed \emph{after} sending its input value to~$a$.
In this world $w_2$, it is the case that agent~$a$ knows that the input of $b$ was $2$.
Hence, the atomic proposition $\inputprop{b}{2}$ must be true in $w_2$, even though the agent $b$ is dead.
\item The worlds, $w_4$, $w_7$ and $w_{10}$ represent situations where one agent died before being able to send any message.
Thus, it is as if only two agents have ever existed, and the labelling only encodes the corresponding two local states:  $\ell(w_4) = \{\inputprop{a}{1}, \inputprop{b}{2}\}$, 
$\ell(w_7) = \{\inputprop{b}{2}, \inputprop{c}{3}\}$ and $\ell(w_{10}) = \{\inputprop{a}{1}, \inputprop{c}{3}\}$.
\item Similarly, $w_0$, $w_{11}$ and $w_{12}$ have labelling $\{\inputprop{a}{1}\}$, $\{\inputprop{b}{2}\}$ and $\{\inputprop{c}{3}\}$ respectively. 
\end{itemize}
We will see in \cref{ex:semantics} some formulas that are true or false in this simplicial model.
\end{example}

\subsection{Semantics of epistemic logic}

Partial epistemic models are a special case of the usual Kripke models; so we can straightforwardly define the semantics of an epistemic formula $\phi\in \lang_K$ in these models.
Formally, gven a pointed partial epistemic model $(\model,w)$,
we define by induction on~$\phi$ the \emph{satisfaction relation} $\model,w \models \phi$
which stands for ``in the world $w$ of the model $\model$, it holds that $\phi$''.
\[ \begin{array}{lcl}
\model,w \models p & \text{iff} & p \in L(w) \\
\model,w \models \neg\phi & \text{iff} & \model,w \not\models \phi \\
\model,w \models \phi\et\psi & \text{iff} & \model,w \models \phi \text{ and } \model,w \models \psi \\
\model,w \models K_a \phi & \text{iff} & \model,w' \models \phi \text{ for all } w' \text{ such that } w \sim_a w'
\end{array} \]

We now take advantage of the equivalence with simplicial models (\cref{thm:equiv2}) to define the interpretation of a formula $\phi\in \lang_K(A,P)$ in a simplicial model.
Given a pointed simplicial model $(\C,X)$ where $X \in \cF(C)$ is a facet of $\C$, we define the relation $\C,X \models \phi$ by induction:
\[ \begin{array}{lcl}
\C,X \models p & \text{iff} & p \in \ell(X) \\
\C,X \models \neg\phi & \text{iff} & \C,X \not\models \phi \\
\C,X \models \phi\et\psi & \text{iff} & \C,X \models \phi \text{ and } \C,X \models \psi \\
\C,X \models K_a \phi & \text{iff} & \C,Y \models \phi \text{ for all } Y \in \cF(C) \text{ such that } a \in \chi(X \inter Y)
\end{array} \]


\begin{example}
\label{ex:semantics}
In the simplicial model of \cref{ex:simpmodel}, we have, for instance:
\begin{itemize}
    \item In world $w_1$, agent $a$ knows the values of all three agents, i.e.\ $\C,w_1 \models K_a (\inputprop{a}{1} \wedge \inputprop{b}{2} \wedge \inputprop{c}{3})$ since $w_2$ and $w_3$ are indistinguishable from $w_1$ by agent $a$ and $\inputprop{a}{1} \wedge \inputprop{b}{2} \wedge \inputprop{c}{3}$ is true in these three facets. This corresponds to the view of process~$a$, see Example \ref{ex:simpmodel0}. 
    \item In $w_3$, agent $a$ knows the values of all three agents but agent $b$ only knows the values of $a$ and $b$: $\C,w_3 \models K_a (\inputprop{a}{1} \wedge \inputprop{b}{2} \wedge \inputprop{c}{3})$ but $\C,w_3 \models K_b (\inputprop{a}{1} \wedge \inputprop{b}{2})$ and $\C,w_3 \models \neg K_b\, \inputprop{c}{3}$ since in facet $w_4$ do not have $\inputprop{c}{3}$.
    Similarly, in $w_4$, agents $a$ and $b$ know each other's values, but do not know the input value of agent $c$: $\C,w_4 \models K_a (\inputprop{a}{1} \wedge \inputprop{b}{2})$, $\C,w_4 \models K_b (\inputprop{a}{1} \wedge \inputprop{b}{2})$, $\C,w_4 \models (\neg K_a\, \inputprop{c}{3})  \wedge (\neg K_b\, \inputprop{c}{3})$
    \item 
    In world $w_1$, agent $a$ knows that agent $b$ knows about their respective input values: 
    $\C,w_1 \models K_a K_b (\inputprop{a}{1} \wedge \inputprop{b}{2})$ but agent $a$ does not know if agent $b$ knows about the value of agent $c$: $\C,w_1 \models \neg K_a K_b\, \inputprop{c}{3}$ (because of $w_3$).
\end{itemize}
\end{example}

\noindent
As expected, our two interpretation of $\lang_K$ agree up to the equivalence of \cref{thm:equiv2}: 

\begin{proposition} \label{prop:truth}
Given a pointed simplicial model $(\C,X)$, $\C,X \models \varphi$ iff $\kappa(\C,X) \models \varphi$.
Conversely, given a pointed proper partial epistemic model $(M,w)$, $M,w \models \varphi$ iff $\sigma(M,w) \models \varphi$.
\end{proposition}
This is straightforward by induction on the structure of the formula $\varphi$.

\subsection{Reasoning about alive and dead agents}
\label{sec:aliveanddead}

In \cref{ex:semantics}, we only considered formulas talking about what the agents know about each other's input values.
It is a natural idea to also contemplate formulas expressing which agents are alive or dead, for example ``agent $a$ knows that agent $b$ is dead''.
Fortunately, such formulas can already be expressed in our logic without any extra work, as derived operators 
$\deadprop{a} \,:=\, K_a\, \false$, and
$\aliveprop{a} \,:=\, \neg \deadprop{a}$. 
It is easy to check that indeed: 
\begin{itemize}
\item \makebox[4.5cm]{In partial epistemic models,\hfill} $\model,w \models \aliveprop{a} \quad \text{iff} \quad w \sim_a w$.
\item \makebox[4.5cm]{In simplicial models,\hfill} $\C,X \models \aliveprop{a} \quad \text{iff} \quad a \in \chi(X)$.
\end{itemize}

\begin{example}
\label{ex:alive}
Consider again the simplicial model of Examples \ref{ex:simpmodel0} and \ref{ex:simpmodel}, and its corresponding partial epistemic model of Example \ref{ex:frame}. It is easy to see that:
\begin{itemize}
\item $M,w_3 \models \aliveprop{b}\wedge \aliveprop{a}$ but $M,w_3 \models \deadprop{c}$, 
    \item $M, w_1 \models \neg K_a\, \aliveprop{c}$ since e.g. $M, w_3 \models \deadprop{c}$ whereas $M, w_1 \models \aliveprop{c}$,
    \item Agents $a$ and $b$ know, in world $w_4$, that $c$ is dead: $M,w_4 \models K_b\,\deadprop{c} \wedge K_a\,\deadprop{c}$ since, first, in world $w_3$ (which is indistinguishable from $w_3$ by agent $b$), agent $c$ is not alive, and second, in world $w_5$ (which is indistinguishable from $w_3$ by agent $a$, $c$ is not alive either. 
\end{itemize}
In $w_4$ everything looks as if agents $a$ and $b$ were executing solo, without $c$ ever existing, whereas in worlds $w_3$ and $w_5$, agent $c$ dies at some point, but has been active and its local value has been observed by one of the other agents. 
\end{example}

\subsection{The axiom system $\KBfour$}

\label{sec:axiomsystem}

We consider the usual proof theory of normal modal logics, with all propositional tautologies, closure by modus ponens, and the necessitation rule: if $\phi$ is a tautology, then $K_a \phi$ is a tautology.
In normal modal logics, there is a well-known correspondence between properties of Kripke models that we consider, and corresponding axioms that make the logic sound and complete~\cite{sep-logic-modal}.
In our case, partial epistemic models are symmetric and transitive. Thus we get the logic $\KBfour$, obeying the following additional axioms.
\begin{align*}
& \textbf{K}: K_a (\phi \Rightarrow \psi) \implies (K_a \phi \Rightarrow K_a \psi)\\
& \textbf{B}: \phi \implies K_a \neg K_a \neg \phi \\ 
& \textbf{4}: K_a \phi \implies K_a K_a \phi
\end{align*}
%
The difference between $\KBfour$ and the more standard multi-agent epistemic logics $\Sfive$ is that we do 
not necessarily have axiom \textbf{T}: $K_a \phi \implies \phi$.
Axiom \textbf{T} is valid in Kripke models whose accessibility relation is reflexive, which we do not enforce.
The logic $\KBfour$ is in fact equivalent to $\mathbf{KB45_n}$ (see e.g.~\cite{sep-logic-modal}), so we also have for free the Axiom \textbf{5}, which corresponds to Euclidean Kripke frames.
We have the following well-known result, see e.g.~\cite{fagin}.

\begin{theorem}
\label{thm:completenessKB4}
The axiom system $\KBfour$ is sound and complete with respect to the class of partial epistemic models.
\end{theorem}

\noindent
Here are a few examples of valid formulas in $\KBfour$ (proofs are in \cref{app:appA}), related to the liveness of agents.
\begin{itemize}
\item \makebox[5.5cm]{Dead agents know everything:\hfill} $\KBfour \vdash \deadprop{a} \implies K_a \phi$.
\item \makebox[5.5cm]{Alive agents know they are alive:\hfill} $\KBfour \vdash \aliveprop{a} \implies K_a\, \aliveprop{a}$.
\item \makebox[5.5cm]{Alive agents satisfy Axiom \textbf{T}:\hfill} $\KBfour \vdash \aliveprop{a} \implies (K_a \phi \Rightarrow \phi)$.
\item \makebox[5.5cm]{Only alive agents matter for $K_a \phi$:\hfill} $\KBfour \vdash K_a \phi \iff (\aliveprop{a} \Rightarrow K_a \phi)$.
\end{itemize}

As an application of the fourth tautology, notice that a formula of the form $K_a K_b \phi$ is equivalent to $K_a (\aliveprop{b} \Rightarrow K_b \phi)$.
So, to check whether this formula is true in some pointed model $(M,w)$, we only need to check that $K_b \phi$ is true in the worlds $w' \sim_a w$ where $b$ is alive.

\subsection{Completeness for simplicial models}
\label{sec:completeness}

According to \cref{thm:equiv2}, simplicial models are equivalent to \emph{proper} partial epistemic models. 
Thus \cref{thm:completenessKB4} does not apply directly to simplicial models, and some extra care must be taken to deal with this ``proper'' requirement.
Indeed, it is easy to check that the two formulas below are true in every simplicial model; but they are not provable in $\KBfour$.
\begin{align*}
& \makebox[1cm]{$\mathbf{NE}$:\hfill} \textstyle \bigvee_{a \in A} \aliveprop{a}\\
& \makebox[1cm]{$\mathbf{SA_a}$:\hfill} \textstyle \left(\aliveprop{a} \land \bigwedge_{b \neq a} \deadprop{b}\right) \implies (\phi \implies K_a\,\phi)
\end{align*}
The formula $\mathbf{NE}$ (Non-Emptiness) says that in every world, there is at least one agent that is alive; and the formula $\mathbf{SA_a}$ (Single Agent) says that if there is exactly one alive agent~$a$, this agent knows everything that is true about the world.
It is straightforward to check that:

\begin{proposition}
\label{cor:35}
The axiom system $\KBfour+\mathbf{NE}+(\mathbf{SA_a})_{a \in A}$ is sound with respect to the class of simplicial models.
\end{proposition}
\begin{proof}
Let us first consider axiom \makebox[1cm]{$\mathbf{NE}$:\hfill} $\bigvee_{a \in A} \aliveprop{a}$. 
Take a proper epistemic model $M=\langle W,\sim\rangle$. 
To prove that for all $w \in W$, $M, w \models \mathbf{NE}$, we have to prove that there exists $a \in A$ such that $w \sim_a w$. This is by definition of properness.

We now turn to axiom \makebox[1cm]{$\mathbf{SA_a}$:\hfill} $\left(\aliveprop{a} \land \bigwedge_{b \neq a} \deadprop{b}\right) \implies (\phi \implies K_a\,\phi)$. Take $M$ again, a proper epistemic frame, and $w \in W$ such that $M, w \models \aliveprop{a} \land \bigwedge_{b \neq a} \deadprop{b}$. We must prove that, assuming $M, w \models \phi$, we have 
$M, w \models K_a \phi$. 
As $M, w \models \aliveprop{a} \land \bigwedge_{b \neq a} \deadprop{b}$, $w \sim_a w$ and, for all $b \neq a$, there is no $w'$ such that $w\sim_b w'$. 

Consider now any $u$ such that $w \sim_a u$, we need to show that $M,u \models \phi$. But in $w$, only $a$ is alive, and by the properness property of $M$, such a $u$ is necessarily equal to $w$. This is because if $u \neq w$, it has to be distinguished by some agent that is alive in $w$, which can only be $a$ by hypothesis on $w$, which contradicts the fact that $w \sim_a u$. Therefore we trivially have $M, u \models \phi$ since $M,w \models \phi$ by assumption. 
\end{proof}

We also believe that this axiom system is complete; but the proof is more involved and we leave it for future work.
We provide a proof sketch below.

\begin{conjecture}
\label{conj:completeness}
$\KBfour+\mathbf{NE}+(\mathbf{SA_a})_{a \in A}$ is complete w.r.t.\ the class of simplicial models.
\end{conjecture}
\begin{proof}[Proof sketch]
We prove completeness for the class of proper partial epistemic models.
Completeness for simplicial models then follows directly by \cref{prop:truth}.
%
%
As usual in completeness proofs, we build a canonical model $\Mc$ whose worlds are maximal and consistent sets of formulas (for the logic $\KBfour+\mathbf{NE}+(\mathbf{SA_a})_{a \in A}$).
The usual machinery (Lindenbaum's Lemma, the Truth Lemma) works as expected.

All we have to do to complete the proof is show that $\Mc$ is a proper partial epistemic model.
Showing that $\Mc$ is a partial epistemic model is standard (see e.g.~\cite{sep-logic-modal}): the axioms $\mathbf{B}$ and $\mathbf{4}$ are used to prove symmetry and transitivity, respectively.
However, the model $\Mc$ is in fact not proper: while the axiom $\mathbf{NE}$ ensures that every world has at least one alive agent, non-proper behaviour (such as the one of \cref{ex:proper}) can occur within $\Mc$.

To fix this, we resort to the classic \emph{unwinding} construction.
From $\Mc$, we build an unwinded model $U(\Mc)$ whose worlds are paths in $\Mc$, of the form $(w_0, a_1, w_1, \ldots, a_k, w_k)$, where each $w_i$ is a world of $\Mc$ and for all $i$, $w_i \sim_{a_{i+1}} w_{i+1}$.
This model $U(\Mc)$ can be shown to be bisimilar to $\Mc$. Moreover, $U(\Mc)$ is proper: behaviours such as the one of \cref{ex:proper} are ruled out by the unwinding construction.
The only remaining possibility for non-properness concerns worlds with a unique agent; they are ruled out by the axioms~$\mathbf{SA_a}$.
\end{proof}

The axioms $\mathbf{NE}$ and $\mathbf{SA_a}$ embody the ``hidden'' assumptions in the use of simplicial models.
Note that we could easily get rid of $\mathbf{NE}$ by allowing the existence of a fictitious $(-1)$-dimensional simplex representing an empty world. This is known in geometry as \emph{augmented} simplicial complexes.
However, the axioms $\mathbf{SA_a}$ are more substantial, and reflect usual assumptions of distributed computing modelling.

\subsection{Knowledge gain}
\label{sec:knowledge-gain}

In~\cite{gandalf-journal}, a key property of the logic used in distributed computing applications is the so-called ``knowledge gain'' property.
This principle says that agents cannot acquire new knowledge along morphisms of simplicial models. Namely, what is known in the image of a morphism was already known in the domain.
The knowledge gain property is used when we want to prove that a certain simplicial map~$f : \C \to \D$ cannot exist.
To achieve this, we choose a formula~$\phi$ and show 
that $\phi$ is true in every world of~$\D$,
and that $\phi$ is false in at least one world of~$\C$.
Then by the knowledge gain property, the map $f$ does not exist. Such a formula $\phi$ is called a \emph{logical obstruction}.
While we are not interested in proving distributed computing results in this paper (the synchronous crash model of \cref{fig:broadcast} is merely an illustrative example), we still check that some version of the knowledge gain property holds, as a sanity check towards future work.


The knowledge gain property that appeared in~\cite{gandalf-journal} applied to \emph{positive} epistemic formulas, i.e., they are cannot talk about what an agent does not know.
Here, we also require an additional condition, which says that every atomic proposition $p \in \AP_a$ that appears in the formula must be \emph{guarded} by a conditional making sure that agent~$a$ is alive.
This is because there might be agents that are dead in the domain of a morphism, but are alive in the codomain.

Formally, the fragment of \emph{guarded positive epistemic formulas} $\phi \in \PL{K,\text{alive}}$ is defined by the grammar
$
\varphi ::= \aliveprop{B} \Rightarrow \psi \mid \varphi \land \varphi \mid \varphi \lor \varphi \mid 
K_a\varphi, \ a \in A,\; B \subseteq A,\; \psi \in \PropL{B}
$, 
where the formula $\aliveprop{B}$ stands for $\bigwedge_{a \in B} \aliveprop{a}$, and the formula $\psi \in \PropL{B}$ is a \emph{propositional formula restricted to the agents in $B$}, defined formally by the grammar:
$
\psi ::= p \mid \neg \psi \mid \psi \land \psi, \ p \in \AP_B
$.

\begin{theorem}[knowledge gain]
\label{thm:lose-knowledge-3}
Consider simplicial models $\C=\la V,S,\chi,\ell \ra$ and
$\D = \la V',S',\chi',\ell' \ra$, and a  morphism of pointed simplicial models $f : (\C,X) \to (\D,Y)$.
Let $\varphi \in \PL{K,\textup{alive}}$ be a guarded positive epistemic formula.
Then $\D,Y \models \varphi$ implies $\C,X \models \varphi$.
\end{theorem}
\begin{proof}
We proceed by induction on the structure of the guarded positive formula $\varphi$.

For the base case, assume $\phi = \aliveprop{B} \Rightarrow \psi$ for some set of agents $B \subseteq A$ and some propositional formula $\psi \in \PropL{B}$.
We distinguish two cases.
Either some agent $a \in B$ is dead in the world $X$, in which case $\C,X \models \phi$ is true.
Or all agents in $B$ are alive in $X$, and since $f(X) \subseteq Y$ (because $f$ is a morphism of pointed simplicial models), all agents in $B$ are also alive in $Y$.
Thus, we have $\D,Y \models \psi$.
Moreover, since $f$ is a morphism, we know that $\ell(X) \cap \AP_{\chi(X)} = \ell(Y) \cap \AP_{\chi(X)}$.
In particular, this yields $\ell(X) \cap \AP_{B} = \ell(Y) \cap \AP_{B}$ because $B \subseteq \chi(X)$.
So all atomic propositions in $\AP_{B}$ have the truth value in the worlds $X$ and $Y$.
As a consequence $\D,Y \models \psi$ implies that $\C,X \models \psi$, and thus $\C,X \models \phi$ as required.

The cases of conjunction and disjunction follow trivially from the induction hypothesis.
Finally, for the case of a formula $K_a \phi$, suppose that $\D,Y \models K_a \varphi$. If $a \not\in \chi(X)$ then $\C,X \models K_a \varphi$, trivially (dead agents know everything).
So let us assume that $a \in \chi(X)$.
In order to show $\C,X \models K_a \varphi$, assume that $a \in \chi(X \cap X')$ for some facet $X'$,
and let us prove $\C,X' \models \varphi$.
Let $v$ be the $a$-coloured vertex in $X \cap X'$.
Then $f(v) \in f(X) \cap f(X')$.
Recall that $f(X) \subseteq Y$ by assumption, and let $Y'$ be a facet of $\D$ containing $f(X')$.
So $f(v) \in Y \cap Y'$, and since $\chi(f(v)) = a$, we get $a \in \chi(Y \cap Y')$ and thus $\D, Y' \models \varphi$.
By induction hypothesis, we obtain $\C,X' \models \varphi$.
\end{proof}

\section{Conclusion}

We began exposing the interplay between epistemic logics and combinatorial geometry in~\cite{gandalf-journal}.  The importance of this
perspective has been well established in distributed computing, where the topology of the simplicial model determines the solvability of a distributed task~\cite{herlihyetal:2013}.
Here we extended it to situations where agents may die: impure simplicial complexes need to be considered.
Many technical interesting issues arise, which shed light on the epistemic assumptions hiding behind the use of simplicial models.

But the main point is that our work opens the way to give a formal epistemic semantics to distributed systems where processes may fail and  failures are detectable (as in the synchronous crash failure model).
It would be interesting to use our simplicial model to  reason about the solvability of tasks in such systems, for example, the following have not been studied using epistemic logic, to the best of our knowledge:  non-complete communication (instead of broadcast situation we considered here) graphs~\cite{CastanedaFPRRT19}, and tasks such as renaming~\cite{okun2010} and lattice agreement~\cite{ZhengG20}.
Especially interesting would be  extending the set agreement logical obstruction of~\cite{yagiNishimura2020TR} to the synchronous crash setting.

Finally, we hope that our simplicial semantics can be useful to reason not only about distributed computing, but also about in other situations with interactions beyond pairs of agents~\cite{BATTISTON20201}.
For instance, impure simplicial complexes have been shown to occur when modelling social systems, neuroscience, and other biological systems (see e.g.~\cite{mock2021political}).

\bibliography{bibliography.bib}

\newpage

\appendix

\section{Proofs of the sample valid formulas in $\KBfour$, Section \ref{sec:axiomsystem}}
\label{app:appA}

\begin{proof}
We begin by proving that $\KBfour \vdash \deadprop{a} \Rightarrow K_a \phi$. 
By the $\mathbf{K}$ axiom, we have $K_a(\false \Rightarrow \phi)\Rightarrow (K_a \false \Rightarrow K_a \phi)$. But $\false \Rightarrow \phi$ is a tautology, and by the necessitation rule, $K_a(\false \Rightarrow \phi)$ is a tautology. Hence $K_a \false \Rightarrow K_a \phi$ but $\deadprop{a}\equiv K_a \false$. 

\medskip

We then prove that $\KBfour \vdash \aliveprop{a} \Rightarrow K_a\, \aliveprop{a}$.
By axiom $\mathbf{B}$ we know that $\true \Rightarrow K_a \neg K_a \false$, that is, $\true \Rightarrow K_a \aliveprop{a}$, hence $K_a \aliveprop{a}$. As a matter of fact, either $a$ is dead and it knows everything by the first property above, even $K_a \aliveprop{a}$ or $a$ is alive, and knows it is alive. 

\medskip

Now we prove that $\KBfour \vdash \aliveprop{a} \Rightarrow (K_a \phi \Rightarrow \phi)$.
We will show the contrapositive, $\KBfour \vdash (K_a\,\phi \land \neg \phi) \Rightarrow \deadprop{a}$.
Assume $K_a\,\phi$ and $\neg \phi$, we want to show $\deadprop{a}$, i.e.\ $K_a\,\false$.
By axiom~\textbf{B}, $\neg \phi \Rightarrow K_a \neg K_a\,\phi$; so by modus ponens, $K_a \neg K_a\,\phi$.
Moreover, by axiom~\textbf{4} and the assumption of $K_a\,\phi$, we get $K_a K_a\,\phi$.
Therefore, since we proved both $K_a \neg K_a\,\phi$ and $K_a K_a\,\phi$, by axiom~\textbf{K} and modus ponens, we obtain $K_a \false$.
  
\medskip

Finally we prove that $\KBfour \vdash K_a \phi \iff (\aliveprop{a} \Rightarrow K_a \phi)$. The left to right implication is trivial. Now suppose $\aliveprop{a} \Rightarrow K_a \phi$, we want to prove that $K_a \phi$. By modus ponens $\deadprop{a} \vee \aliveprop{a}$ and if $\deadprop{a}$ then $a$ knows everything by the first property we proved, for instance $K_a \phi$. If $\aliveprop{a}$ then, because $\aliveprop{a} \Rightarrow K_a \phi$, $K_a \phi$ holds. 
\end{proof}

\end{document}